\documentclass[10pt]{amsart}

\usepackage{amssymb,amsthm,amsmath,enumerate}
\usepackage[numbers,sort&compress]{natbib}
\usepackage{color}
\usepackage{graphicx}
\usepackage{amssymb,amsthm,amsmath}
\usepackage[numbers,sort&compress]{natbib}
\usepackage{color}
\usepackage{graphicx}

\title[ ]{ The asymptotical behaviour of   embedded eigenvalues  for  perturbed periodic operators}

\author{Wencai Liu}
\address[Wencai Liu]{Department of Mathematics, University of California, Irvine, California 92697-3875, USA}\email{liuwencai1226@gmail.com}

\theoremstyle{plain}
\newtheorem{theorem}{Theorem}[section]
\newtheorem{corollary}[theorem]{Corollary}
\newtheorem{lemma}[theorem]{Lemma}
\newtheorem{proposition}[theorem]{Proposition}

\newcommand{\R}{\mathbb{R}}

\newcommand{\N}{\mathbb{N}}
\theoremstyle{definition}

\newtheorem{remark}[theorem]{Remark}

\begin{document}


\begin{abstract}
Let $H_0$ be a periodic operator on $\R^+$(or periodic Jacobi operator on $\N$). It is known that the absolutely continuous spectrum  of $H_0$ is consisted of spectral bands $\cup[\alpha_l,\beta_l]$.
Under the assumption  that $\limsup_{x\to \infty} x|V(x)|<\infty$ ($\limsup_{n\to \infty} n|V(n)|<\infty$), the
  asymptotical behaviour of  embedded  eigenvalues approaching to the spectral boundary  is established.

\end{abstract}
\maketitle
\section{Introduction and main results}
We   consider  the continuous Schr\"odinger  equation on $  \R$ ($\R^+$),
\begin{equation}\label{Gu}
    Hu=-u^{\prime\prime}+(V(x)+V_0(x))u=Eu,
\end{equation}
where $V_0(x)$ is  1-periodic and $V (x)$ is a decaying  perturbation.

In this paper, we always assume $V_0$ is 1-periodic and in $L^1[0,1]$. Without loss of generality, we  only consider the half-line equation.

When $V\equiv 0$, we have a $1$-periodic Schr\"odinger equation,
\begin{equation}\label{GV}
    H_0\varphi=-\varphi^{\prime\prime}+V_0(x)\varphi=E\varphi.
\end{equation}
It is known that the absolutely continuous  spectrum (essential spectrum) of $H_0$   is
consisted  of a union of   closed
intervals (often referred to as  spectral  bands). We denote by
\begin{equation*}
  \sigma_{\rm ac}(H_0) = \sigma_{\rm ess}(H_0)=\bigcup_{l}[\alpha_l,\beta_l].
\end{equation*}
The bands may collapse at some boundaries. Putting the collapsed  bands together, we denote by
\begin{equation*}
  \sigma_{\rm ac}(H_0) =  \sigma_{\rm ess}(H_0)=\bigcup_{l}[\tilde{\alpha}_l,\tilde{\beta}_l].
\end{equation*}
The last band may have  the form  as $[c,\infty)$.
In order to make the difference, we call $[\alpha_l,\beta_l]$ a {\it standard spectral band } and $[\tilde{\alpha}_l,\tilde{\beta}_l]$ a {\it non-standard spectral band}.

By Floquet theory, for any $E\in(\alpha_l,\beta_l)$, there exists  $\varphi$ of \eqref{GV}, which has the following form
\begin{equation}\label{Gfqp}
    \varphi(x,E)=p(x,E)e^{i k(E)x}
\end{equation}
 where $k(E)$ is the quasimomentum, and $p(x,E)$ is 1-periodic.
In each standard spectral  band $[\alpha_l,\beta_l]$, $k(E)$  is monotonically increasing from $0$ to $\pi$ or monotonically decreasing from $\pi$ to $0$.

Similarly,
we  consider the periodic Jacobi equation,
\begin{equation}\label{eq:JacobiUnperturbed}
  (H_0u)(n):=  a_{n+1} u({n+1})+a_n u({n-1})+b_{n+1}u(n)=Eu(n), n\geq 0,
\end{equation}
 where the $\{a_j, b_j\}$ are real sequences with  $q$-period, and
   $a_j$ is   positive for all possible $j$.

Likewise, we  consider a class of the  perturbed periodic Jacobi equation,
 \begin{equation}\label{eq:JacobiPerturbed}
  (Hu)(n)=  a_{n+1} u({n+1})+a_n u({n-1})+(b_{n+1}+V({n+1}))u(n)=Eu(n), n\geq 0,
 \end{equation}
 where  $V(n)$ is an real sequences. Now we will extend the notations for the continuous case to the discrete one.
 For   $E\in (\alpha_l,\beta_l)$, let $\varphi$  be the Floquet solution of \eqref{eq:JacobiUnperturbed}, namely,
\begin{equation}\label{Gflo}
 \varphi(n,E)= p(n,E) e^{i  \frac{k(E)}{q}n},
 \end{equation}
 where $p(n,E)$ is a $q$-periodic sequence  and $k(E)\in(0,\pi)$ is   the quasimomentum ($q$ is the period for $a_n$ and $b_n$).
   We   also have the standard bands and  non-standard bands notations,
 \begin{equation*}
   \sigma_{\rm ac}(H_0) = \sigma_{\rm ess}(H_0)=\cup_{l}[ {\alpha}_l, {\beta}_l]=\cup_{l}[\tilde{\alpha}_l,\tilde{\beta}_l].
 \end{equation*}
 It is a very interesting problem to study  phenomenons of   eigenvalues/singular continuous spectrum embedded into the absolutely continuous  (or essential) spectrum based on the decaying perturbations, for example
\cite{sim07,kissc,kru,lotoreichik2014spectral,Naboko2018,lukwn,luk14,lukd1,simonov2016zeroes,KRS,KLS,MR2945209,remlingsharp,killip02imrn,lastlukic}.
We care about  the asymptotical behavior of   eigenvalues approaching to the spectral boundaries.
The  asymptotical behavior of eigenvalues lying outside spectral bands has  been well studied
\cite{sa,dk04,dc07,damanik2003variational,killip2003sum}.

Our goal is to investigate   the asymptotical behavior of eigenvalues approaching  to the spectral boundaries from the inside  part.
In our previous note, we show that there is no embedded eigenvalues in the spectral band if $V(x)=\frac{o(1)}{1+x}$.
If we allow $V(x)=\frac{h(x)}{1+x}$ for  some positive function $h(x)$ going to $\infty$, $H_0+V$ can have the desired countably  embedded eigenvalues \cite{ldcombine}.
In  those two cases, it is trivial to study the asymptotical behavior  of embedded eigenvalues.
The purpose of this paper is to study the remaining case, namely $V(x)=\frac{O(1)}{1+x}$.
If $V_0=0$ (it is reduced to the free case), both the continuous  and discrete cases have  been studied in \cite{KLS} and \cite{liu2018criteria} respectively.
Starting from the Floquet solutions  for periodic operators, we obtained a universal and optimal  result for  the asymptotical behavior.
It leads to various applications. See Remark \ref{re} for the details. In our forthcoming papers, we will show that there is no singular continuous spectrum if $V(x)=\frac{O(1)}{1+x}$ \cite{sc1,sc2}.

In order to state our main results,   more notations are necessary.
 The Wronskian for the discrete case is given by
 \begin{equation*}
   W(f,g)(n)=f(n)g(n+1)-f(n+1)g(n).
 \end{equation*}
The  Wronskian for the continuous case is given by
 \begin{equation*}
   W(f,g)(x)=f(x)g^{\prime}(x)-f^{\prime}(x)g(x).
   \end{equation*}
 Since the Floquet solution $\varphi(x,E)$ ($\varphi(n,E)$) is a solution of $H_0u=Eu$, the Wronskian $W(\overline{\varphi},\varphi)$ is constant and there is an $\omega\in\R$ 
 such that
\begin{equation*}
 W(\overline{\varphi},\varphi)=i\omega.
\end{equation*}
 In the continuous case,  define
\begin{equation*}
  \Gamma(E)=\int_0^1\frac{4}{\omega^2}|\varphi(x,E)|^4dx.
\end{equation*}
 In the discrete case,  define
 \begin{equation*}
   \Gamma(E)=\frac{1}{q}\sum_{n=1}^{q}\frac{4}{\omega^2}|\varphi(n,E)|^4.
 \end{equation*}
 $\Gamma(E)$ is well defined for $E\in(\alpha_l,\beta_l)$ and blows up as $E\to \alpha_l$ or $E\to \beta_l$.
 In the following arguments, $H_0$ is always  a  continuous periodic operator or periodic Jacobi operator. $V(x)$ or $V(n)$ is the perturbed potential.


In the continuous case, define $P$ as
\begin{equation}\label{Gdef.Pc}
  P=\{E\in \R:-u^{\prime\prime}+(V(x)+V_0(x))u=Eu \text { has an }L^2(\R^+) \text { solution }\},
\end{equation}
and in the discrete case,
define $P$ as
\begin{equation}\label{Gdef.Pd}
  P=\{E\in \R:  \text{ equation } \eqref{eq:JacobiPerturbed}  \text { has an }\ell^2(\N) \text { solution }\}.
\end{equation}
$P$ is the collection of    eigenvalues  for operators $H_0+V$ with all the possible boundary conditions at $0$.

 \begin{theorem}\label{Thmunic}

Suppose $\limsup_{x\to \infty}x|V(x)|= A$ ($\limsup_{n\to \infty}n|V(n)|= A$).
  For any standard spectral band $[\alpha_l,\beta_l]$, let $\delta_l$ be the unique point in $[\alpha_l,\beta_l]$ such that the  quasimomentum $k(\delta_l)=\frac{\pi}{2}$. Then   the set $P\cap (\alpha_l,\beta_l)$ is a countable set with two possible accumulation points $\alpha_l$ and $\beta_l$. Moreover,
 \begin{equation}\label{Gapr12Gamma}
    \sum_{E_i\in P\cap (\alpha_l,\delta_l)} \frac{1}{\Gamma(E_i) }\leq \frac{A^2}{2},
 \end{equation}
 and
 \begin{equation*}
     \sum_{E_i\in P\cap (\delta_l,\beta_l)} \frac{1}{\Gamma(E_i) }\leq \frac{A^2}{2}.
 \end{equation*}

\end{theorem}

  For any $\lambda\in\cup_l\{\alpha_l,\beta_l\}$, we  define $\kappa_{\lambda}=1$ if $\lambda$ is non-collapsed. Otherwise, $\kappa_{\lambda}=2$.
  See the formal description  of ``non-collapsed" in Theorem \ref{ThmNonc} and the Remark after.

Based on Theorems \ref{Thmunic}, we can get a lot of Corollaries.
 \begin{corollary}\label{thm1}
 Suppose $\limsup _{x\to \infty}x|V(x)|=A$ ($\limsup _{n\to\infty}n|V(n)|=A$). Then for any standard spectral band $[\alpha_l,\beta_l]$,
 $P\cap[\alpha_l,\beta_l] $
  is a countable  set. Moreover, there is a constant $K$ (depends on $k$ uniformly in any bounded set ) such that
 \begin{equation*}
    \sum_{E_i\in P\cap (\alpha_l,\beta_l)}\min\{|E_i-\alpha_l|^{\kappa_{\alpha_l}},|E_i-\beta_l|^{\kappa_{\beta_l}}\}\leq KA^2.
 \end{equation*}
  In particular,
 \begin{equation*}
    \sum_{E_i\in P\cap (\alpha_l,\beta_l)}\min\{|E_i-\alpha_l|^{2},|E_i-\beta_l|^{2}\}\leq K A^2.
 \end{equation*}
 \end{corollary}
 \begin{corollary}\label{thm2}

 Suppose $\limsup_{ x\to\infty}x|V(x)|=A$ ($\limsup_{n\to\infty} n|V(n)|=A$). Then for any non-standard spectral band $[\tilde{\alpha}_l,\tilde{\beta}_l]$, there exists some $\epsilon$ (only depends on $H_0$ not on $l$)   and  a constant $K$ (depends on $l$ uniformly in  any bounded set) such that
 \begin{equation*}
    \sum_{E_i\in P\cap (\tilde{\alpha}_l,\tilde{\alpha}_l+\epsilon)} \min\{|E_i-\tilde{\alpha}_l|,|E_i-\tilde{\beta}_l|\}\leq K A^2,
 \end{equation*}
 and
 \begin{equation*}
    \sum_{E_i\in P\cap (\tilde{\beta}_l-\epsilon,\tilde{\beta}_l)} \min\{|E_i-\tilde{\alpha}_l|,|E_i-\tilde{\beta}_l|\}\leq K A^2.
 \end{equation*}
 \end{corollary}
 \begin{remark}
 In the continuous case, $\beta_l-\alpha_l=2\pi^2 l +O(1)$ as $k$ goes to infinity \cite{eastham1973spectral}. Thus we can choose  $ \epsilon \asymp\pi^2 l$ for large $l$.
 \end{remark}
\begin{corollary}\label{cor1}
Suppose $V_0(x)=0$ and $\limsup _{x\to\infty}x|V(x)|=A$. Then
    $P\cap (0,\infty)$ is a countable set. Moreover,
 \begin{equation*}
    \sum_{E_i\in P\cap (0,\infty)}E_i\leq \frac{ A^2}{2}.
 \end{equation*}
\end{corollary}
\begin{corollary}\label{cor2}
Suppose $V_0(n)=0$ and $\limsup_{n\to\infty} n|V(n)|=A$. Then
    $  P\cap (-2,2)$ is a countable set. Moreover,
 \begin{equation*}
    \sum_{E_i\in P\cap (-2,2)}(4-E_i^2)\leq 4 A^2+4\min\{1,A\}.
 \end{equation*}
\end{corollary}
For the discrete case, $H_0$ has $q$ standard spectral bands and the set of spectrum is bounded. Then the constant $K$ ($\epsilon$) only depends on $H_0$.

\begin{corollary}\label{cor3}
 Suppose  $\limsup_{n\to \infty} n|V(n)|=A$. Let $d(E)=\min_{l}\text{dist}(E,\{\alpha_l,\beta_l\})$. Then   there is a constant $K$   such that
 \begin{equation*}
    \sum_{E_i\in P \cap \sigma_{\rm ess}(H_0)} d^2(E_i)\leq K A^2.
 \end{equation*}
 \end{corollary}
 \begin{corollary}\label{cor4}
 Suppose  $\limsup_{n\to\infty} n|V(n)|=A$. Then there exist   constants $\epsilon>0$ and $K>0$  such that
 \begin{equation*}
    \sum_{E_i }\min\{|E_i-\tilde{\alpha}_l|,|E_i-\tilde{\beta}_l|\}\leq K A^2,
 \end{equation*}
 where $\{E_i\}$  goes over all the values in set   $P\cap \left(\cup_l(\tilde{\alpha}_l,\tilde{\alpha}_l+\epsilon)\cup(\tilde{\beta}_l-\epsilon,\tilde{\beta}_l)\right)$.
 \end{corollary}
\begin{remark} \label{re}

\begin{itemize}
\item All the bounds go to 0 as $A$ goes to zero. It implies that there is no $L^2(\R^+)$ (or $\ell^2(\N)$) solution of $(H_0+V)u=Eu$  for $E\in \cup_l(\alpha_l,
\beta_l)$ if  $V(x)=\frac{o(1)}{1+x}$
($V(n)=\frac{o(1)}{1+n}$). Thus under any fixed boundary condition at $0$,  there is no eigenvalues of the operator  $H_0+V$  embedded  into its  bands if  $V(x)=\frac{o(1)}{1+x}$
($V(n)=\frac{o(1)}{1+n}$).
\item Corollary \ref{cor1} (Corollary \ref{cor2}) with exactly the same bound has been proved in \cite{KLS} (\cite{liu2018criteria}). It means our bounds in Theorem \ref{Thmunic}  are  effective.
    \item  Fix $l$. Corollary \ref{thm1} implies the speed of  $E_i\in P$ going to the collapsed boundaries behaves  like $|E_i-\alpha_l|\asymp\frac{1}{1+\sqrt{i}}$ ($|E_i-\beta_l|\asymp\frac{1}{1+\sqrt{i}}$). The speed going to non-collapsed boundaries behaves like
        $|E_i-\alpha_l|\asymp\frac{1}{1+ {i}}$ ($|E_i-\beta_l|\asymp\frac{1}{1+{i}}$).
        Remling \cite{remling2000schrodinger} shows that we can not improve it to $|E_i-\alpha_l|\asymp\frac{1}{1+ {i}^{1+\epsilon}}$ in the continuous case. This means that our main results are optimal in some sense.
\end{itemize}
\end{remark}
Our paper is organized in the following way. In Section \ref{SecFlo}, we will introduce the Floquet solution and generalized Pr\"ufer transformation.
In Section \ref{Seccon}, we will give the proof of  all the results in the    continuous setting. In Section \ref{Secdis}, we will give the proof of  all the results in the   discrete setting.

In the following, $E$ is always in some spectral band  $(\alpha_l,\beta_l)$. The constant $C$   depends on $E$, and so does  $O(1) $.
$K$  depends on $E$ (or say $l$) uniformly  in any compact set.
\section{Floquet solution and generalized Pr\"ufer transformation for the continuous Schr\"odinger operator}\label{SecFlo}
Let $T_0(E)$ be the transfer matrix of $H_0$ from $0$ to $1$, that is
\begin{equation*}
  T_0(E)\left(\begin{array}{cc}
                u(0) \\ u^\prime(0)
              \end{array}
  \right)=\left(\begin{array}{cc}
                u(1) \\ u^\prime(1)
              \end{array}
  \right)
\end{equation*}
for any solution $u$ of $H_0u=Eu$.

Let
\begin{equation*}
  T_0(E)=\left(
           \begin{array}{cc}
             a(E) & b(E) \\
             c(E) & d(E) \\
           \end{array}
         \right)
\end{equation*}
and $D(E)=a(E)+d(E)$
\begin{theorem}\cite[Theorem 2.3.3]{eastham1973spectral}\label{ThmNonc}
Let $\Lambda=\cup_l\{\alpha_l,\beta_l\}$.
Then for any $\lambda\in \Lambda$, $D(\lambda)$ is either $2$ or $-2$.  Furthermore, one of  $D^\prime{(\lambda)}$ and $D^{\prime\prime}{(\lambda)}$
must be non-zero.
\end{theorem}
\begin{remark}
$D^\prime{(\lambda)}\neq 0$ if and only if $\lambda$ is non-collapsed.
\end{remark}
 For $E\in (\alpha_l,\beta_l)$, let $k(E)\in (0,\pi)$ be  the quasimomentum.
Let $\varphi(x,E)$ be a solution of $H_0u=Eu$ with $E\in (\alpha_l,\beta_l)$ and satisfies boundary condition
$\varphi(0)=-b(E)$ and $\varphi^{\prime}(0)=a(E)-e^{ik(E)}$. It is easy to check that
\begin{equation*}
  T_0(E)
   \left(
           \begin{array}{c}
             \varphi(0)  \\
             \varphi^{\prime}(0)
           \end{array}
         \right)=e^{ik(E)} \left(
           \begin{array}{c}
             \varphi(0) \\
              \varphi^{\prime}(0)
           \end{array}\right).
\end{equation*}
We get a  Floquet solution $\varphi(x,E)$.

Since $\varphi(x,E)$ is a solution of $H_0u=Eu$, the Wronskian $W(\overline{\varphi},\varphi)$ is constant and
\begin{equation*}
 W(\overline{\varphi},\varphi)=2i \text{ Im } (\overline{\varphi}\varphi^{\prime})
\end{equation*}
Let $\omega(E)\in \R$ be such that
\begin{equation*}
  2i \text{ Im } (\overline{\varphi(x,E)}\varphi^{\prime}(x,E))=i\omega(E).
\end{equation*}
Thus
\begin{equation*}
  \omega(E)=2b(E)\sin k(E).
\end{equation*}
Without loss of generality, assume   $\omega(E)>0$ (i.e., $b(E)>0$).

Let $\gamma(x,E)$ be such that
\begin{equation*}
  \varphi(x,E) =|\varphi(x,E)|e^{i \gamma(x,E)}.
\end{equation*}
Then (See \cite[Prop.2.1]{KRS})
\begin{equation}\label{Ggamma}
  \gamma^{\prime}(x,E)=\frac{\omega(E)}{2|\varphi(x,E)|^2}.
\end{equation}
\begin{proposition}[ Proposition 2.2 and Theorem 2.3bc of \cite{KRS}] \label{thmformula}
Suppose $u$ is a real solution of \eqref{Gu}.
Then there exist real functions $ R(x,E)>0$ and $\theta(x,E)$ such that
\begin{equation}\label{GformulaR}
    [\ln R(x,E)]^\prime=\frac{V(x)}{2\gamma^\prime(x,E)} \sin2 \theta(x,E)
\end{equation}
and
\begin{equation}\label{Gformulatheta}
   \theta(x,E)^\prime=\gamma^\prime(x,E)-\frac{V(x)}{2\gamma^\prime(x,E)} \sin^2 \theta(x,E).
\end{equation}
Moreover,  there exists a constant $C$(depending on $E$) such that
\begin{equation}\label{Gformulanorm}
    \frac{|u(x,E)|^2+|u^\prime(x,E)|^2}{C} \leq R(x,E)^2\leq C (|u(x,E)|^2+|u^\prime(x,E)|^2).
\end{equation}

\end{proposition}
\begin{lemma}
The following estimate holds,
\begin{equation}\label{Gapr111}
    \frac{1}{\gamma^{\prime}(x,E)}  \leq \frac{K(E)}{\sin k(E)},
\end{equation}
and then (in the continuous case)
\begin{equation*}
  \Gamma(E)\leq \frac{K(E)}{|\sin k(E)|^2},
\end{equation*}
where $K(E)$ is uniformly bounded on any compact set of $E$.
\end{lemma}
\begin{proof}
By \eqref{Ggamma}, it suffices to show that
\begin{equation}\label{Gapr112}
    \frac{|\varphi(x,E)| ^2}{\omega} \leq \frac{K(E)}{\sin k(E)}.
\end{equation}
Since $|\varphi(x,E)|$ is 1-periodic, it suffices to prove  \eqref{Gapr112} for $0\leq x\leq 1$.
By the fact that $V_0\in L^1[0,1]$ and Gr\"onwall's inequality, one has
\begin{equation*}
  |\varphi(x,E)|^2\leq K(E)(|\varphi(0,E)|^2+|\varphi^{\prime}(0,E)|^2).
\end{equation*}
Thus it suffices to show
\begin{equation*}
    \frac{|\varphi(0,E)|^2+|\varphi^{\prime}(0,E)|^2}{\omega} \leq \frac{K(E)}{\sin k(E)}.
\end{equation*}
By the definition,
\begin{equation*}
     \frac{|\varphi(0,E)|^2+|\varphi^{\prime}(0,E)|^2}{\omega}=\frac{b^2(E)+(a(E)-\cos k(E))^2+\sin^2 k(E)}{2b\sin k(E)}.
\end{equation*}
Using $a(E)+d(E)=2\cos k(E)$ and $a(E)d(E)-b(E)c(E)=1$, one has
\begin{eqnarray*}
  b^2(E)+(a(E)-\cos k(E))^2+\sin^2 k(E) &=& b^2(E)+a^2(E)-2a(E)\cos k(E)+1 \\
  &=& b^2(E)-a(E)d(E)+1\\
  &=& b^2(E)-b(E)c(E).
\end{eqnarray*}
Since $||T_0||\leq K(E)$, one has $|b(E)|\leq K, |c(E)|\leq K$. It implies
\begin{equation*}
    \frac{|\varphi(0,E)|^2+|\varphi^{\prime}(0,E)|^2}{\omega}=\frac{b(E)-c(E)}{2\sin k(E)}\leq \frac{K(E)}{\sin k(E)}.
\end{equation*}
We finish the proof.

\end{proof}
 \begin{lemma}\cite[Lemma 4.4]{KLS}\label{Leapr7}
  Let $\{e_i\}_{i=1}^N$ be a set of unit vectors  in a Hilbert space $\mathcal{H}$  so that
  \begin{equation*}
    \alpha=N\sup_{i\neq j}| \langle e_i,e_j\rangle|<1.
  \end{equation*}
  Then
  \begin{equation}\label{Gapr71}
    \sum_{i=1}^N|\langle g,e_i\rangle|^2\leq (1+\alpha)||g||^2.
  \end{equation}
  \end{lemma}

\section{Proof of the results in continuous settings }\label{Seccon}
Although the following Lemma is proved for  a special class of  functions in \cite{ldcombine}, it works for all $L^2[0,1]$ functions.
 \begin{lemma}\label{Lcon1}
Suppose $V(x)=\frac{O(1)}{1+x}$ and $f\in L^2([0,1])$ is a 1-periodic function.
 Then
\begin{equation}\label{Gcons4}
  \left|  \int _{0}^x f(t)\frac{\cos 4 \theta(t,E)}{1+t}dt\right |= O(1),
\end{equation}
for  $E \in\cup_l(\alpha_l,\beta_l)$ with $k(E)\neq \frac{\pi}{2}$.
 Suppose  $E_1,E_2\in \cup_l(\alpha_l,\beta_l)$ satisfy  $k(E_1)\neq k(E_2)$ and $k(E_1)+k(E_2)\neq \pi$. Then we have
\begin{equation}\label{Gcons5}
    \left| \int_{0}^x f(t)\frac{\sin 2 \theta(t,E_1) \sin 2 \theta(t,E_2)}{1+t}dt\right |=O(1).
\end{equation}

\end{lemma}
 \begin{proof}[\bf Proof of Theorem \ref{Thmunic} in the continuous setting]
By the assumption of Theorem \ref{Thmunic}, for any $M>A$, we have
\begin{equation*}
 | V(x)|\leq \frac{M}{1+x}
\end{equation*}
for large $x$.
By shifting the equation, we can assume
\begin{equation}\label{Gapr76}
  | V(x)|\leq \frac{M}{1+x}
\end{equation}
for all $x>0$.

Let us take any standard spectral band $(\alpha,\beta)=(\alpha_l,\beta_l)$ into consideration.
Let $\delta\in (\alpha,\beta)$ be  such that the quasimomentum $k(\delta)=\frac{\pi}{2}$.

  Suppose  $E_1,E_2,\cdots E_N\in(\alpha,\delta)\cap P$.
  By \eqref{Gformulanorm} (also use $u\in L^2$ implies $u^{\prime}\in L^2$), we  have
  \begin{equation*}
    \sum_{i=1}^NR(x,E_i)\in L^2(\R^+),
  \end{equation*}
  and  then there exists $B_j\to \infty$ such that
  \begin{equation}\label{Gapr72}
    R (B_j,E_i) \leq  B_j^{-\frac{1}{2}},
  \end{equation}
   for all $i=1,2,\cdots,N$.

   By \eqref{GformulaR}, one has
\begin{equation}\label{GPrufRmar14}
   \ln R(x,E) - \ln R(0,E)=\int_0^x\frac{V(t)}{2\gamma^\prime(t,E)} \sin2 \theta(t,E) dt.
\end{equation}
  By \eqref{Gapr72} and \eqref{GPrufRmar14}, we have
  \begin{equation}\label{Gapr75}
\int_0^{B_j}\frac{V(t)}{\gamma^\prime(t,E_i)} \sin2 \theta(t,E_i) dt \leq -B_j+O(1),
  \end{equation}
   for all $i=1,2,\cdots,N$.

Now let us consider the Hilbert spaces
\begin{equation*}
  \mathcal{H}_j=L^2((0,B_j),(1+x)dx).
\end{equation*}

In  $\mathcal{H}_j$, by \eqref{Gapr76} we have
\begin{equation}\label{Gapr77}
  ||V||_{ \mathcal{H}_j}^2\leq M^2\log (1+B_j).
\end{equation}
Let
\begin{equation*}
  e^j_{i}(x)=\frac{1}{\sqrt{A_i^j}}\frac{ \sin 2\theta(x,E_i)}{\gamma^{\prime}(x,E_i)(1+x)}\chi_{[0,B_j]}(x),
\end{equation*}
where $A_i^j$ is chosen so that $e_i^j$ is a unit vector in $\mathcal{H}_j$.

Obviously,
\begin{eqnarray}
  A_i^j &=& \int_{0}^{B_j}\frac{\sin^2 2\theta(x,k_i)}{|\gamma^{\prime}(x,E_i)|^2(1+x)}dx \nonumber\\
   &=&\int_{0}^{B_j}\frac{1}{2|\gamma^{\prime}(x,E_i)|^2(1+x)}dx - \int_{0}^{B_j}\frac{\cos 4\theta(x,k_i)}{|\gamma^{\prime}(x,E_i)|^2(1+x)}dx \label{Gapr791}.
\end{eqnarray}

By \eqref{Gcons4}, one has
\begin{equation}\label{Gcons4Gapr79}
  \left| \int_{0}^{B_j}\frac{\cos 4\theta(x,k_i)}{|\gamma^{\prime}(x,E_i)|^2(1+x)}dx\right|= O(1),
\end{equation}
for all $i=1,2,\cdots,N$.

Direct computation shows that
\begin{eqnarray}
  \int_{0}^{B_j}\frac{1}{|\gamma^{\prime}(x,E_i)|^2(1+x)}dx &=& O(1)+\sum_{n=0}^{B_j-1} \int_{n}^{n+1}\frac{1}{|\gamma^{\prime}(x,E_i)|^2(1+n)}dx\nonumber\\
  &=& O(1)+\Gamma(E_i)\log B_j .\label{Gamma1}
\end{eqnarray}
By \eqref{Gapr791}, \eqref{Gcons4Gapr79} and \eqref{Gamma1}, we have
\begin{equation}\label{Gapr79}
      A_i^j= \frac{1}{2}\Gamma(E_i)\log B_j +O(1).
\end{equation}
Since $E_i\in(\alpha,\delta)$  ($k(E_i)\in (0,\frac{1}{2})$) for all $i=1,2,\cdots,N$, one has
  \begin{equation*}
    k_i+k_{g}\neq \pi,
  \end{equation*}
  for all $1\leq i,g\leq N$.

By  \eqref{Gcons5}, we have
\begin{equation}\label{Gapr78}
 \frac{O(1)}{\log B_j}  \leq \langle e_i^j,e_{ g}^j \rangle\leq\frac{O(1)}{\log B_j},
\end{equation}
for  all $1\leq i,g\leq N$ and $i\neq g$.

By \eqref{Gapr79} and \eqref{Gapr75}
\begin{equation}\label{Gapr710}
 \langle V,e^j_i \rangle_{\mathcal{H}_j}\leq -\frac{\sqrt{2}}{\sqrt{\Gamma(E_i)}} \sqrt{\log B_j}+O(1),
\end{equation}
for large $j$.
By \eqref{Gapr71} and \eqref{Gapr78}, one has
\begin{equation}\label{Gapr711}
\sum_{i=1}^N  |\langle V,e^j_i\rangle_{\mathcal{H}_j}|^2\leq (1+\frac{O(1)}{\log B_j}||V||_{\mathcal{H}_j}.
\end{equation}
By \eqref{Gapr710}, \eqref{Gapr711} and \eqref{Gapr77},
we have
\begin{equation*}
  \sum_{i=1}^N \frac{2}{\Gamma(E_i)}  \log B_j\leq   M^2 \log  B_j+O(1).
\end{equation*}
Let $j\to \infty$, we get
\begin{equation*}
  \sum_{i=1}^N \frac{1}{\Gamma(E_i)}  \leq  \frac{M^2}{2},
\end{equation*}
for any $M>A$. This implies
\begin{equation}\label{Gapr122}
   \sum_{i=1}^N \frac{1}{\Gamma(E_i)}  \leq  \frac{ A^2}{2} .
\end{equation}
By \eqref{Gapr111}, one has
\begin{equation*}
  \Gamma(E)\leq \frac{K(E)}{|\sin k(E)|^2}.
\end{equation*}
It implies
\begin{equation}\label{Gapr121}
   \sum_{i=1}^N |\sin k(E_i)|^2  \leq  K A^2 .
\end{equation}
$(\alpha,\delta)\cap P$ is a  countable set   with  possible accumulation   point  $\alpha$. Similarly,
 $(\delta,\beta)\cap P$ is a  countable set   with  possible accumulation  point  $\beta$, and the bounds as
in \eqref{Gapr122} and \eqref{Gapr121} hold. We finish the proof.
\end{proof}
We need another lemma to treat the special situation $E=\delta$ separately.
\begin{lemma}\label{Leapr12c}\cite{ldcombine}
Fix any spectral band $[\alpha_l,\beta_l]$. Let $\delta_l\in[\alpha_l,\beta_l]$ be such that $k(\delta_l)=\frac{\pi}{2}$. Then
there exists some $K>0$ such that  if $\limsup_{x\to \infty}x|V(x)|\leq \frac{1}{K}$, then $\delta_l\notin P$.
\end{lemma}

\begin{proof}[\bf Proof of Corollary \ref{thm1} in the continuous setting]

By \eqref{Gapr121} and also taking   $(\delta,\beta)\cap P$ into consideration, we have
\begin{equation}\label{Gapr123new}
   \sum_{E_i\in P\cap (\alpha,\beta), E_i\neq \delta} |\sin k(E_i)|^2  \leq  K A^2 .
\end{equation}
Combining with Lemma \ref{Leapr12c}, we have,
\begin{equation}\label{Gapr123}
   \sum_{E_i\in P\cap (\alpha,\beta)} |\sin k(E_i)|^2  \leq  K A^2 .
\end{equation}
By the fact that $2\cos \pi k(E_i)=D(E_i)$,  one has
\begin{equation}\label{Gapr102}
  \sum_{E_i\in P\cap (\alpha,\beta)}  (4-D(E_i)^2)   \leq  KA^2 .
\end{equation}
Without loss of generality, assume   $D(\alpha)=2$ and $D(\beta)=-2$. If $\alpha$ is non-collapsed,
\begin{equation}\label{Gapr124}
    2-D(E_i)\geq \frac{1}{K}|E_i-\alpha|.
\end{equation}
If $\alpha$ is collapsed, by Theorem \ref{ThmNonc}, one has
\begin{equation}\label{Gapr125}
    2-D(E_i)\geq \frac{1}{K}|E_i-\alpha|^2.
\end{equation}
Similarly,
 if $\beta$ is non-collapsed,
\begin{equation}\label{Gapr124new}
    2+D(E_i)\geq \frac{1}{K}|E_i-\beta|.
\end{equation}
If $\beta$ is collapsed, by Theorem \ref{ThmNonc}, one has
\begin{equation}\label{Gapr125new}
    2+D(E_i)\geq \frac{1}{K}|E_i-\beta|^2.
\end{equation}

By \eqref{Gapr123}, \eqref{Gapr124}, \eqref{Gapr125}, \eqref{Gapr124new} and \eqref{Gapr125new}, we have
\begin{equation}\label{Gapr103}
  \sum_{E_i\in P\cap (\alpha,\beta)}\min\{ |E_i-\alpha|^{\kappa_{\alpha}} ,|E_i-\beta|^{\kappa_{\beta}} \}\leq  KA^2 .
\end{equation}

\end{proof}
\begin{proof}[\bf Proof of Corollary \ref{thm2} in the continuous setting]
Since $ \tilde{\alpha}_l$ is non-collapsed, there exists some standard spectral band $[\alpha,\beta]\subset [\tilde{\alpha}_l,\beta_l]$ such that $\alpha=\tilde{\alpha}$ and
$D^{\prime}(\alpha)\neq 0$.
Let $\epsilon=\frac{\beta-\alpha}{2}$. Obviously, for  any   $E\in(\alpha,\alpha+\epsilon)\cap P$,
\begin{equation*}
  |E-\beta|\geq \frac{1}{K} .
\end{equation*}
Then
\begin{equation}\label{Gapr128}
  \min\{ |E-\beta|,|E-\alpha|\}\geq \frac{|E-\alpha|}{K}.
\end{equation}
Notice that $\kappa_{\alpha}=1$.
By \eqref{Gapr103} and \eqref{Gapr128}, we have
\begin{equation*}
    \sum_{E_i\in P\cap (\tilde{\alpha}_l,\tilde{\alpha}_l+\epsilon]} |E_i-\tilde{\alpha}_l| \leq K A^2.
 \end{equation*}
 Similarly,
 \begin{equation*}
    \sum_{E_i\in P\cap (\tilde{\beta}_l-\epsilon,\tilde{\beta}_l)} \min\{|E_i-\tilde{\beta}_l|\}\leq K A^2.
 \end{equation*}

\end{proof}
\begin{proof}[\bf Proof of Corollary \ref{cor1}]
For the free Schr\"odinger equation, we can choose Floquet solution $\varphi(x,E)=e^{i\sqrt{E}x}$ for $E>0$.
Direct computation and by \eqref{Ggamma}, one has
\begin{equation}\label{Ggammafree}
  \gamma^{\prime}(x,E)=\sqrt{E}.
\end{equation}
and then
\begin{equation}\label{Ggammafree}
  \Gamma(E)=\frac{1}{ {E}}.
\end{equation}
In this case,   let $\alpha=0$, $\beta=\infty$ and $\delta=\infty$.
Following the proof of  \eqref{Gapr122}, we have
\begin{equation*}
  \sum_{E_i\in P\cap (0,\infty)}E_i\leq \frac{A^2}{2}.
\end{equation*}
\end{proof}
\section{Proof of the  results in  discrete settings }\label{Secdis}

In this section, all the equations  are discrete.
Similarly, we can also introduce the Pr\"ufer variables $R(n,E)$ and $\theta(n,E)$.
See \cite{KRS} for details.
\begin{equation*}
  \frac{R(n+1,E)^2}{R(n,E)^2}=1-V(n+1)\frac{2}{\omega}\sin2\theta(n,E)\vert \varphi(n,E)\vert^2+\frac{4V(n+1)^2\vert \varphi(n,E)\vert^4}{\omega^2}\sin^2\theta(n,E).
\end{equation*}

Let us take one standard spectral band $[\alpha,\beta]=[\alpha_l,\beta_l]$ into consideration.
Let $\delta\in (\alpha,\beta)$ be  such that the quasimomentum $k(\delta)=\frac{\pi}{2}$.
  \begin{lemma}\cite{ldcombine}\label{Lcon1d}
Suppose $V(n)=\frac{O(1)}{1+n}$ and $f(n) $ is $q$ periodic. Let $E_1,E_2\in (\alpha,\beta)$ be such that $k(E_1)\neq k(E_2)$ and $k(E_1)+k(E_2)\neq \pi$.
Assume $k(E_1)\neq \frac{\pi}{2}$.
Then  for any $\varepsilon>0$, there exist $D(E_1,E_2,\varepsilon)$ and $D(E_1,\varepsilon)$  such that
\begin{equation}\label{Gcons4d}
  \left|  \sum _{t=1}^n f(t)\frac{\cos 4 \theta(t,E_1)}{1+t}\right|\leq D(E_1,\varepsilon)+ {\varepsilon}\ln n,
\end{equation}
and
\begin{equation}\label{Gcons5d}
    \left|\sum_{t=1}^nf(t)\frac{\sin 2 \theta(t,E_1) \sin 2 \theta(t,E_2)}{1+t}\right|\leq D(E_1,E_2,\varepsilon)+{\varepsilon}\ln  n.
\end{equation}

\end{lemma}


\begin{proof}[\bf Proof of Theorem \ref{Thmunic} in the discrete setting]
Replacing Lemma   \ref{Lcon1} by Lemma  \ref{Lcon1d},
the proof of  discrete case   can be proceeded in a similar way as that of continuous case. We omit the details.
\end{proof}

\begin{lemma}\label{Leapr12d}\cite{ldcombine}
There exists some $K>0$ such that  if $\limsup_{n\to \infty}n|V(n)|\leq \frac{1}{K}$, then $\delta\notin P$.
\end{lemma}

\begin{proof}[\bf Proof of Corollaries \ref{thm1} and \ref{thm2} in the discrete setting]
The proof follows from the continuous case. We omit the details here.
\end{proof}
In order to prove Corollary \ref{cor2}, we also need one Lemma.
\begin{lemma}\label{Leapr12d1}\cite[Theorem 2.1]{liu2018criteria}
 If $V_0=0$ and $\limsup_{n\to \infty}n|V(n)|< 1$, then $0$ can not be in $P$.
\end{lemma}
\begin{proof}[\bf Proof of Corollary \ref{cor2}]

In this case, we only have one spectral band $[-2,2]$. Moreover, $\varphi(n,E)=e^{ik(E)n}$, $2\cos k(E)=E$ for $E\in(-2,2)$ and $\Gamma(E)=\frac{1}{\sin^2 k(E)}$.
By Theorem \ref{Thmunic}($\delta=0$),
\begin{equation*}
    \sum_{E_i\in P\cap(-2,2);E_i\neq 0} \frac{1}{\Gamma(E_i)}\leq A^2.
 \end{equation*}
 Thus
\begin{equation}\label{Gapr104}
  \sum_{E_i\in P\cap(-2,2);E_i\neq 0}(4-E_i^2)   \leq  4A^2 .
\end{equation}
We may add 4 in the bound of \eqref{Gapr104} if  0 is in $P$.
However, if $ A<1$, by Lemma \ref{Leapr12d}, 0 can not be in $P$.
We finish the proof.
\end{proof}
\begin{proof}[\bf Proof of Corollaries \ref{cor3} and \ref{cor4}]
It follows from Corollaries  \ref{thm1} and \ref{thm2} immediately since the discrete operator has finitely many spectral bands.
\end{proof}
 \section*{Acknowledgments}
    This research was   supported by  NSF DMS-1401204 and  NSF DMS-1700314.

\end{document}